\documentclass[10pt,journal,letterpaper,compsoc]{IEEEtran}

\usepackage{graphicx}
\usepackage{subfigure}

\usepackage[linesnumbered,ruled,vlined]{algorithm2e}
\usepackage{amsfonts}
\usepackage{amsmath,bm}
\usepackage{amssymb}
\usepackage{amsthm}
\usepackage{multirow}
\usepackage{xcolor}
\usepackage{setspace}
\usepackage{comment}
\usepackage{cite}
\usepackage{array}
\usepackage{color}
\usepackage{colortbl}
\usepackage{caption}
\usepackage{enumitem}
\usepackage{textcomp}

\newtheorem{theorem}{Theorem}

\newtheorem{definition}{Definition}

\newcommand{\ie}{{\em i.e.}}
\newcommand{\eg}{{\em e.g.}}
\newcommand{\et}{{\em et al.}}
\newcommand{\st}{{\em s.t.}}

\begin{document}

\title{Fine-Grained User Profiling for Personalized Task Matching in Mobile Crowdsensing}


\author{Shuo Yang, \IEEEmembership{Student Member,~IEEE,}
Zhenzhe Zheng, \IEEEmembership{Student Member,~IEEE,}
Shaojie Tang, \IEEEmembership{Member,~IEEE,}
Fan Wu, \IEEEmembership{Member,~IEEE,}
~Guihai~Chen,~\IEEEmembership{Senior Member,~IEEE,}
\IEEEcompsocitemizethanks{\IEEEcompsocthanksitem S. Yang, Z. Zheng, F. Wu, and G. Chen are with the Shanghai Key Laboratory of Scalable Computing and Systems, Department of Computer Science and Engineering, Shanghai Jiao Tong University, China. E-mails: \{wnmmxy, zhenzhezheng\}@sjtu.edu.cn; \{fwu, gchen\}@cs.sjtu.edu.cn.

S. Tang is with the Department of Information Systems, University of Texas at Dallas, USA. E-mail: tangshaojie@gmail.com.
\IEEEcompsocthanksitem F. Wu is the corresponding author.
}
}

\IEEEcompsoctitleabstractindextext{
\begin{abstract}
\renewcommand{\raggedright}{\leftskip=0pt \rightskip=0pt plus 0cm}\raggedright
In mobile crowdsensing, finding the best match between tasks and users is crucial to ensure both the quality and effectiveness of a crowdsensing system. Existing works usually assume a centralized task assignment by the crowdsensing platform, without addressing the need of fine-grained personalized task matching. In this paper, we argue that it is essential to match tasks to users based on a careful characterization of both the users' preference and reliability. To that end, we propose a personalized task recommender system for mobile crowdsensing, which recommends tasks to users based on a recommendation score that jointly takes each user's preference and reliability into consideration. We first present a hybrid preference metric to characterize users' preference by exploiting their implicit feedback. Then, to profile users' reliability levels, we formalize the problem as a semi-supervised learning model, and propose an efficient block coordinate descent algorithm to solve the problem. For some tasks that lack users' historical information, we further propose a matrix factorization method to infer the users' reliability levels on those tasks. We conduct extensive experiments to evaluate the performance of our system, and the evaluation results demonstrate that our system can achieve superior performance to the benchmarks in both user profiling and personalized task recommendation.
\end{abstract}
\begin{IEEEkeywords}
Mobile Crowdsensing, Task Matching, User Profiling, Truth Discovery, Recommender System
\end{IEEEkeywords}}

\maketitle

\section{Introduction}
\label{intro}
Due to the rapid development of smart devices and wireless technology, mobile crowdsensing \cite{ganti2011mobile} has risen as an emerging sensing paradigm. It can employ a large number of smart devices to extract and share their local information using the embedded sensors on them. A typical mobile crowdsensing system usually consists of three major components: crowdsensing platform, service requesters, and mobile device users (data contributors). The platform is responsible for handling information requests from the service requesters and publishing sensing tasks to the users through the interaction of their smartphone applications.

A critical problem in crowdsensing is to find the best match between users and tasks. Most of the existing works adopt a \emph{platform-centric model} \cite{yang2012crowdsourcing,zhao2014crowdsource,zhao2014fair,karaliopoulos2015user,zhang2016incentive,jin2015quality}, which allows the platform to make centralized decisions on which users are selected to perform which sensing tasks. These works usually focus on the incentive problem, where a typical procedure goes like this: each user submits a bid reflecting her willingness or cost in participating in a task, and then the platform determines the set of selected users and their payments, so as to optimize certain utility metric (\eg, coverage, revenue, service quality) and satisfy some game-theoretic properties. The underlying assumption behind this type of model is that the users are fully rational and are capable of determining their optimal strategies. However, as pointed out in \cite{karaliopoulos2016first}, this assumption, as well as the setting that each user's preference can be abstracted as a single bidding parameter, could be an oversimplification of the complicated user behaviors.

Another type of task matching systems, referred as \emph{user-centric model}, gives the users more freedom to choose their interested tasks. It has been widely adopted in many commercial crowdsensing systems, such as Waze \cite{waze}, Field Agent \cite{fieldagent}, and Gigwalk \cite{gigwalk}. In these systems, the available tasks are shown to the users via their smartphone applications. The users can manually browse through the task corpus (often with simple built-in filters, such as proximity filter and payment filter), and choose their interested tasks to participate in. However, since the number of the tasks is often really large, it is inefficient for the users to browse page by page searching for suitable tasks. Without an efficient personalized task matching solution, the users may end up selecting tasks that they are not familiar with or not interested in, which may result in a decrement of the quality of their collected sensing data.

Considering the limitations of existing task matching works, we propose to design a personalized task recommender system for mobile crowdsensing, so as to facilitate the match of the users with suitable tasks. Note that in traditional recommender systems, such as movie recommendation, items are recommended based only on customers' preferences \cite{adomavicius2005toward}. Whereas, in mobile crowdsensing, besides the metric of the users' preferences, we also need to take the users' reliability/data quality into consideration. That is because the users may have heterogeneous sensing behaviors towards different tasks, which could influence the quality of their collected data \cite{huang2010you}. Achieving preference- and quality-aware task recommendation can have a positive impact on both attracting the user's further participation and improving the crowdsensing system's effectiveness. However, such a personalized task recommender system is missing in the current crowdsensing literature. Jin \et \cite{jin2015quality} and Wang \et \cite{wang2016quality} studied the quality-aware incentive mechanism design without addressing the need of personalized task recommendation. Karaliopoulos \et \cite{karaliopoulos2016first} proposed to assign the tasks to the users based on the profile of each user's probability of accepting a task, but did not consider the users' reliability information.

Central to the personalized task recommender system is a careful characterization on each user's preference and reliability towards different tasks. However, it is not a trivial task, due to the unique nature of the crowdsensing scenarios. One of the challenges is finding a good way to model the users' preference over different tasks. In some traditional recommendation scenarios, customers' preference can be readily obtained from their previous ratings \cite{adomavicius2005toward}. However, the users in mobile crowdsensing do not typically provide explicit ratings on their preference, \st, we have to infer the users' preference from their implicit feedback, including their task browsing history and task selection record.

The most challenging part is estimating the users' reliability levels. In particular, we have to learn the users' reliability information for different tasks based on their submitted sensing data, if any, so as to build each user a profile characterizing the trustworthiness of the users' data for performing the tasks. Although \emph{truth discovery} algorithms \cite{yin2008truth} can be adopted to jointly estimate the users' data quality and the underlying truths, they cannot fully address the need of user reliability profiling in the context of task recommendation. Note that truth discovery algorithms usually generate a single reliability parameter for each user representing the overall trustworthiness level of the user. However, to conduct personalized task recommendation, the heterogeneity of a user's reliability in different tasks has to be exploited, and thus a more fine-grained reliability profiling of the users should be considered. A possible alternative is to independently generate each user a reliability parameter for each task by applying truth discovery algorithms to the data of each sensing task. Unfortunately, this approach may suffer from scalability issue, and what's worse, a user's reliability for a task cannot be estimated by truth discovery algorithms, if the user did not contribute data to that task. This could often be a problem in real crowdsensing scenarios, especially when the users' data are sparse, \ie, each user only contributes data to only a small number of the tasks. Besides, without the prior knowledge of truth and reliability measures, typical truth discovery algorithms are likely to fail, when the majority of data are inaccurate \cite{li2016survey}.

In this work, we jointly consider the problems of user profiling and personalized task matching in mobile crowdsensing, and propose a personalized task recommender system framework, which recommends tasks to the users based on both the users' preference and reliability. We propose approaches to measure the users' preference and reliability, respectively. First, in profiling the users' preferences, we present a hybrid preference metric that integrates the feedback against both the users' historical performance and the preference of their peers. Then, to tackle the more challenging part of profiling the users' reliability, we model the problem as a semi-unsupervised learning problem, and propose an efficient block coordinate descent algorithm to jointly estimate the users' reliability and the unknown ground truths. We surpass existing truth discovery methods by (1) considering grouping tasks into several categories, (2) taking the information of failed tasks into consideration, and (3) using a small number of available truth data to facilitate the estimation accuracy. Note that a user's reliability for certain task category cannot be estimated if the user did not provide data for those tasks. To address this problem, we further propose a matrix factorization method to estimate the missing entries. We conduct a real-world experiment and a large-scale crowdsensing simulation to evaluate the performance of our methods. The evaluation results show that our proposed methods can achieve superior performance over existing works and our benchmarks.

The main contributions of this work are listed as follows.
\begin{itemize}
\item First, we design a personalized task recommender system framework that matches tasks to the users based on both the users' preference and reliability of the tasks. We propose a method to profile each user's preference over the tasks by exploiting the user's implicit feedback.
\item Second, we model the problem of user reliability profiling as a semi-supervised learning model, and propose an efficient algorithm to estimate the users' reliability and the unknown ground truths simultaneously. We also propose a matrix factorization method to estimate each user reliability for tasks she did not contribute data to.
\item Third, we conduct a real-world crowdsensing experiment and a large-scale simulation to evaluate the performance of our methods. Both the experiment and simulation results show that our proposed methods can achieve significant performance improvements to our benchmarks.

\end{itemize}

The rest of the paper is organized as follows. We first present the system overview in Section \ref{model}, and then introduce the problem formulations in Section \ref{formulation}. In Section \ref{algorithm}, we propose our reliability profiling algorithms. We evaluate our proposed methods and present the evaluation results in Section \ref{evaluation}. In Section \ref{related}, we review the related works. Finally, we conclude this paper in Section \ref{discussion}.

\section{System Overview}
\label{model}
In this section, we present an overview of our proposed personalized task recommender system.

\subsection{Personalized Task Recommendation}

Suppose there are $N$ users and $M$ sensing tasks in the system. The set of users and tasks are denoted by $\mathbb{N}$ and $\mathbb{S}$, respectively. We consider a \emph{user-centric model}, where the users can browse the tasks in their smartphone applications and choose to participate in their interested tasks. If a user $i$ wants to participate in a task $j$, she can click on some button to inform the platform her participation. After that, the user will use her smartphone to collect and then submit sensing data to the platform. Let $x_{i,j}$ denote the data submitted by the user $i$ to the task $j$. The ground truth of the task $j$ is denoted by $x_j^*$, which is usually unavailable to the platform.

We tend to build a personalized task recommender system, where the tasks are recommended to the users based on a joint consideration of the users' preference and reliability. Specifically, for each task $j$, suppose each user $i$'s preference and reliability regarding the task is denoted by $p_{i,j}$ and $q_{i,j}$, respectively. We propose a recommendation score $Score(i,j)$ that takes both the user $i$'s preference and reliability for the task $j$ into account, \ie, $Score(i,j) = f(p_{i,j}, q_{i,j})$, where the function $f()$ outputs the recommendation metric based on the two input parameters.

Instances of the function $f()$ can be specified by the platform according to its need. Simple instances may include a linear combination (\ie, $Score(i,j) = \gamma \, p_{i,j} + (1-\gamma) q_{i,j}$) or a product (\ie, $Score(i,j) = p_{i,j} * q_{i,j}$) of the two parameters. One may also treat the personalized recommendation as a constrained optimization problem, where one parameter is applied into the optimization term while the other as the constraint. For example, for each user $i \in \mathbb{N}$, we want to recommend such task $j' \in \mathbb{S}$ that maximizes the user's preference and also satisfies the constraint that the user's reliability for the task should be above a certain threshold. More formally:
\begin{align}
\forall i \in \mathbb{N}, \,\,\,\, & j' \leftarrow \textrm{arg}\, \underset{j \in \mathbb{S}}{\textrm{max}} \,\, p_{i,j}, \notag \\
\textrm{s.t.} \,\,\,\, & q_{i,j'} \geq q_{req},
\end{align}
where $q_{req}$ is the minimum required reliability for a user to perform a task.

Central to the system model is the users' preference and reliability measures. To that end, we need to carefully examine the historical data of the crowdsensing system, in order to acquire profiles of the users' preference and reliability.

\subsection{User Preference Profiling}
\label{sec-pre-profiling}

To characterize the users' preference on the tasks, the users' feedback information is needed. However, in mobile crowdsensing, the users' explicit feedback (\eg, ratings, like or dislike) is usually unavailable. Thus, we have to exploit the users' implicit feedback. Fortunately, the crowdsensing platform can have access to each user's browsing history on the application, including which tasks the user has browsed, selected, and successfully completed. This information can be used to infer the users' preference on the tasks, from two different perspectives, \ie, either against the user's historical performance (content-based methods) or the preferences of other similar users (collaborative methods) \cite{adomavicius2005toward}.

\subsubsection{Content-Based Method}

Each task has many attributes, including time, location, travel distance, payment, category, and so on. Along with the users' task selection choices (selected or not), this information can be regarded as training data. By applying classification methods, such as logistic regression or Bayesian classifier, we can build a classifier for each user to infer her probability of selecting each task \cite{karaliopoulos2016first}. We let $P_i(j)$ denote the probability of the user $i$ selecting the task $j$.

\subsubsection{Collaborative Method}

Let $U$ denote the users' task preference matrix, where the entry $u_{i,j}$ indicates the user $i$'s preference over the task $j$. We assign the value of each $u_{i,j}$ by mapping the user's task browsing history to a task preference value, \ie,
\begin{equation}
u_{i,j} = \left\{
\begin{array}{ll}
\textrm{N/A} & \textrm{if $i$ did not browse task $j$}, \\
0.5 & \textrm{if $i$ browsed but not selected task $j$}, \\
1 & \textrm{if $i$ browsed and selected task $j$}. \\
\end{array}
\right.
\end{equation}
Then, we can apply state-of-the-art collaborative filtering methods to predict these missing entries \cite{koren2009matrix}.

Both of the above two methods have their limitations. On one hand, the content-based method may suffer from an overspecialization problem, \ie, it will only recommend a user tasks that are similar to those she has already selected. On the other hand, the collaborative methods may not perform well when the preference matrix is sparse. To alleviate the limitations of these two methods, we propose a hybrid recommendation approach. To do so, we define each user $i$'s preference for each task $j$ as a linear combination of the content-based characteristic and the collaborative-based characteristic, \ie,
\begin{equation}
q_{i,c} = \eta P_i(j) + (1 - \eta) u_{i,j},
\end{equation}
where $\eta \in [0,1]$ is a hyperparameter.

Many previous works on recommender system have investigated the problem of exploiting customers' implicit feedback in different application contexts. The intuitions of them can be further incorporated to improve our modelling of the users' preferences. Possible extensions may include further considering the users' preference over each category \cite{yuen2015taskrec,bao2012location}, incorporating the implicit negative feedback \cite{carroll1993explicit}, multidimensionality of recommendations \cite{adomavicius2005toward}, and Bayesian personalized ranking \cite{rendle2009bpr}. In the rest of the paper, we tend to put our most efforts on user reliability profiling, which is the most challenging part of the system.

\subsection{User Reliability Profiling}

A user's reliability in performing a sensing task is measured by the quality of her contributed data. Intuitively, if a user's contributed data are accurate, \ie, close to the ground truths, the user would have a higher reliability level, and vice versa. Note that in mobile crowdsensing scenarios, the ground truths are usually unavailable, thus we cannot directly measure the users' reliability level by comparing their data with the ground truths. 

To address this problem, one possible approach is to adopt \emph{truth discovery} algorithms, which are proposed to resolve conflicts in data provided by heterogeneous data sources. Existing truth discovery algorithms, \eg \cite{li2014confidence,li2014resolving,su2014generalized,ouyang2014truth,meng2015truth,yang2017designing}, usually follow a similar unsupervised procedure: first initializing the ground truth estimation using a simple majority voting or averaging scheme, and then iteratively updating reliability and ground truth based on the current estimation of the other. Although truth discovery algorithms have been performed well on many web mining tasks, they cannot be directly applied due to the following unique requirements of our reliability profiling contexts.



\textbf{Multi-dimensional Reliability}: Existing truth discovery algorithms usually output a single reliability parameter for each user characterizing the overall trustworthiness of the user. Whereas, with the objective of personalized task recommendation, differentiation among the tasks is needed, \st, we need to estimate the users' heterogeneous reliability levels for different tasks. Besides, we noticed that in mobile crowdsensing, there are a large number of tasks, where different tasks may require different data collection behaviors, thus a user's reliability may vary towards different tasks. For example, a user who is used to put her smartphone in the bag may fail to provide accurate data in measuring the surrounding noise, while can still have high reliability level in monitoring traffic congestions \cite{yang2017designing}. Thus, a more fine-grained reliability profiling method is needed to character the users' multi-dimensional reliability.

\textbf{Scalability}: One natural idea to address the multi-dimensional reliability problem is to apply existing truth discovery algorithms to each task independently and generate each user $i$ a reliability measure $q_{i,j}$ for each task $j$. However, due to the large number of tasks, calculating a reliability parameter per user per task is not scalable. Besides, estimating each user's reliability based only on her data to a single task may be susceptible to noise, and thus cannot accurately reflect the user's reliability level. Therefore, our reliability profiling method should not only provide a fine-grained reliability estimation, but also be scalable under a large number of tasks.

\textbf{Robustness}: Most truth discovery algorithms start with a uniform initialization of truth values or reliability values. As a result, their performance relies on the assumption that the most users' are reliable. However, when this assumption fails, the iterative computation of truth estimation and reliability estimation may move towards incorrect directions, leading to poor estimation accuracy. This problem, referred as ``initialization problem'' could often occurs in mobile crowdsensing scenarios, due to the uncertainty of each individual human contributor. Thus, it is crucial to design a reliability estimation algorithm that is robust to such scenarios.

\textbf{Complete Reliability Characterization}: Note that the users' reliability values are estimated based on the relative relative accuracy of the their data. In consequence, a user's reliability for certain dimension cannot be estimated if the user did not provide data to that dimension. This may not be a problem in many truth discovery scenarios where their main goal is to infer the unknown ground truths. However, in the context of personalized task recommendation, we have to obtain a complete characterization of the users' reliability levels. Thus, we need to propose a method to predict each user's reliability for those dimensions that the user did not provide data to.

\textbf{Different Data Types}: Different sensing tasks may have different data types. For example, a traffic congestion task may require categorical data (\eg, no congestion, medium congestion, or high congestion), while a noise monitoring task may require continuous numerical data (\ie, the noise levels of the users' surrounding environment). Thus, the reliability profiling algorithm needs to be carefully designed to handle both categorical and continuous data types.

Our proposed user reliability profiling methods are carefully designed to address the above requirements. Specifically, for the multi-dimensional reliability and the scalability issues, we classify tasks into a number of categories and estimate the users' reliability for each category independently. As for the robustness issue, we propose a semi-supervised learning framework that exploits few available truth knowledge to improve the estimation accuracy. We also propose a matrix factorization method to predict the missing entries in reliability estimation. The issue of different data types is taken care of by considering different loss functions. In the subsequent sections, we present the problem formulation and algorithm design of our user reliability profiling problem respectively. 

\section{Problem Formulation}
\label{formulation}
In this section, we formalize the user reliability profiling problem. We first present the problem model, and then propose a preliminary version and two enhancements of our problem. One enhancement is to incorporate the information of failed tasks, and the other is to integrate a small portion of truth data to improve the estimation accuracy.

\subsection{Problem Model}
\label{urp-model}

To model the users' multi-dimensional reliability, we tend to take the similarities among the tasks into consideration by classifying the tasks into different categories, where the tasks within each category focus on a similar sensing target. For example, some category only focuses on noise monitoring tasks, and an other focuses on traffic congestion monitoring. The classification of the tasks is common in current crowdsensing applications, \eg, Waze \cite{waze}. It can be done by the platform's direct designation in the task publication phase, or by applying text classification techniques \cite{forman2003extensive} to automatically analyze the descriptions of the tasks. Specifically, we categorize the $M$ tasks into $C$ categories ($C \ll M$). For each category $c \in \{1, \ldots, C\}$, the set of the tasks belong to the category is denoted by $\mathbb{S}_c$ ($\mathbb{S}_c \subseteq \mathbb{S}$). For simplicity, we assume that each task $j \in \mathbb{S}$ only belongs to one category, thus the sets $\mathbb{S}_1, \ldots, \mathbb{S}_C$ are mutually disjoint. More general situations will be discussed in Section \ref{discussion}. For each task category $c$, let $q_{i,c}$ denote each user $i$'s reliability of the task category. The user reliability profiling problem is to infer the users' reliability for each category. More formally:

\begin{definition}[User Reliability Profiling Problem]
Given a set of users $\mathbb{N}$, a set of interested tasks $\mathbb{S}$, and the users' contributed data $\{x_{i,j}| i \in \mathbb{N}, j \in \mathbb{S}\}$, the user reliability profiling problem aims to estimate the unknown ground truths $\{x_j^* | j\in \mathbb{S}\}$, and the users' reliability matrix $Q \in \mathbb{R}^{N \times C}$, where $C$ is the dimension of each user $i$'s reliability.
\end{definition}

\subsection{Preliminary Problem Formulation}
\label{urp-basic}

We assume that the tasks in different categories are independent, \st, we can estimate the users' reliability for each category separately.
Let $\mathbb{N}_c$ denote the set of users who contributed data to tasks in category $c$. To estimate the users' reliability, for each category $c$, we aim to solve the following optimization problem.
\begin{align}
\underset{\{q_{i,c}\},\{\hat{x}^*_j\}}{\textrm{min}} \,\,\,\, &
\sum_{i \in \mathbb{N}_c} \sum_{j \in \mathbb{S}_c} y_{i,j} \, q_{i,c} \, L(x_{i,j},\hat{x}_j^*), \notag \\
\textrm{s.t.} \,\,\,\, & \delta(\{q_{i,c}\})=1
\label{typical-truth-discovery}
\end{align}
where $y_{i,j}$ indicates if the user $i$ has contributed data to the task $j$, $\hat{x}_j^*$ is our estimation for the task $j$'s ground truth, and $\delta()$ is a regularization function. Following the convention of truth discovery literature \cite{li2014resolving}, we adopt the exponential regularization function, \ie, $\delta(\{q_{i,c}\}) = \sum_{i \in \mathbb{N}_c} \textrm{exp}(-q_{i,c})$. The loss function $L()$ measures the distance between a user's data and the estimated truth. For continuous data, $L()$ can be defined as the squared distance, \ie, $L(x,\hat{x}^*)=(x-\hat{x}^*)^2$, while for categorial data, $L()$ can be defined as the $0/1$ distance, \ie, $L(x,\hat{x}^*)=0$ if $x=\hat{x}^*$, and $1$ otherwise. An intuitive interpretation of the problem formulation is that the ground truth should be close to the data contributed by reliable users, and the users whose data are close to the ground truth should have high reliability levels.

\subsection{Incorporating Information of Failed Tasks}
\label{urp-v1}

We observe that in practice, the users may select certain tasks, but did not successfully complete them (\eg, decide to terminate the sensing procedure half way). The phenomenon, referred as \emph{failed tasks}, is likely to reflect the users' unreliability in performing certain tasks. In this part, we improve the above problem formalization by taking this issue into account.

We first introduce some notations. Among the set of tasks in each category $c$, we let $\mathbb{S}_{i,c}$ denote the set of tasks the user $i$ selected, and $\mathbb{D}_{i,c}$ the set of tasks the user $i$ has successfully completed, where $\mathbb{D}_{i,c} \subseteq \mathbb{S}_{i,c} \subseteq \mathbb{S}_c$. For each category $c$, we calculate each user $i$'s task completion ratio $r_{i,c}$, which is defined as the number of tasks the user $i$ has finished over the number of tasks the user $i$ has selected, \ie, $r_{i,c} = \frac{|\mathbb{D}_{i,c}|}{|\mathbb{S}_{i,c}|}$. We revise the original formulation by multiplying a penalty term to $q_{i,c}$. The revised problem is presented as follows.
\begin{align}
\underset{\{q_{i,c}\},\{\hat{x}^*_j\}}{\textrm{min}} \,\,\,\, &
\sum_{i \in \mathbb{N}_c} \sum_{j \in \mathbb{S}_c} y_{i,j} \, q_{i,c} \, g(r_{i,c})\, L(x_{i,j},\hat{x}_j^*), \notag \\
\textrm{s.t.} \,\,\,\, & \sum_{i \in \mathbb{N}_c} \textrm{exp}(- q_{i,c} \, g(r_{i,c}))=1,
\end{align}
where $g(x)= 1 - \textrm{log}(x)$ is a function mapping each user's completion ratio to a penalty. We can see that the users who have failed tasks will receive a completion ratio less than 1, and thus their reliability outputs should be less than the ones estimated by the previous method shown in Equation \ref{typical-truth-discovery}. An extreme case is that some user $i$ may select multiple tasks but completed zero (\ie, $\mathbb{S}_{i,c}>0$ and $\mathbb{D}_{i,c}=0$). In this case, the system cannot generate a reliability estimation for the user. We will handle this problem in Section \ref{subsection-missing}.

\subsection{Incorporating Available Ground Truths}
\label{urp-v2}


The above formulation extends the basic truth discovery problem, which is built upon an underlying assumption that the majority of data are reliable. Unfortunately, it may suffer from the initialization problem, \ie, when most of the data are unreliable, the above estimation procedure may have bad performance \cite{li2016survey}. To tackle this issue, we propose a \emph{semi-supervised} learning framework, which incorporates a small number of ground truths to improve the estimation accuracy. To this end, the platform may intentionally add a few tasks with known ground truths into the task corpus to collect additional information on the users' reliability, whereas the users have no idea which tasks are inserted by the platform. The platform may also sample a few tasks, and employ some trusted workers to obtain their ground truths. Several heuristic methods can be applied to choose the sampled set of tasks. For example, we may choose the sampled tasks randomly, choose the tasks whose data have the largest variations, or choose the tasks which have the most data contributors.

We let $\mathbb{S}$ denote the set of tasks with unknown ground truths, and $\mathbb{O}$ denote the set of tasks that are intentionally inserted by the platform with known truth information. For each category $c$ of tasks, we let $\mathbb{S}_c$ and $\mathbb{O}_c$ denote the set of the tasks without and with prior ground truths respectively.

Having the ground truths of some tasks in hand, we propose to leverage those information to further enhance our estimation accuracy. To distinguish the notations, we let $\hat{x}^*_j$ denote the estimation of the ground truth ($j \in \mathbb{S}$), and $x^*_o$ denote the known truth ($o \in \mathbb{O}$). Then, for each category $c$, the modified learning optimization problem is given by
\begin{align}
\underset{\{q_{i,c}\},\{\hat{x}^*_j\}}{\textrm{min}} \,\, & \sum_{i \in \mathbb{N}_c} q_{i,c} \, g(r_{i,c}) \Big( \sum_{j \in \mathbb{S}_c} y_{i,j} \, L(x_{i,j}, \hat{x}^*_j) \notag \\
& \,\,\,\,\,\,\,\,\,\,\,\, + \alpha \sum_{o \in \mathbb{O}_c} y_{i,o} \, L(x_{i,o}, x_o^*) \Big), \notag \\
\textrm{s.t.} \,\,\,\, & \sum_{i \in \mathbb{N}_c} \textrm{exp}(- q_{i,c} \, g(r_{i,c}))=1,
\label{semi-final}
\end{align}
where $\alpha$ is a hyper parameter controlling the relative weight of the second loss terms. We can see that the second loss term $\sum_{o \in \mathbb{O}_c} y_{i,o} \, L(x_{i,o}, x_o^*)$ is constant for each user $i$ in each task category $c$. We let $\epsilon_{i,c}$ denote the term $\sum_{o \in \mathbb{O}_c} y_{i,o} \, L(x_{i,o}, x_o^*)$, and the problem presentation can be simplified as follows.
\begin{align}
\underset{\{q_{i,c}\},\{\hat{x}^*_j\}}{\textrm{min}} \,\, & \sum_{i \in \mathbb{N}_c} q_{i,c} \, g(r_{i,c}) \Big( \sum_{j \in \mathbb{S}_c} y_{i,j} \, L(x_{i,j}, \hat{x}^*_j) + \alpha \epsilon_{i,c} \Big) \notag \\
\textrm{s.t.} \,\,\,\, & \sum_{i \in \mathbb{N}_c} \textrm{exp}(- q_{i,c} \, g(r_{i,c}))=1.
\label{semi-final-simple}
\end{align}

We summarize the frequently used notations in Table 1.

\begin{table}[!t] \centering
\caption{Frequently Used Notations}
\label{notation}
\begin{tabular}{cl}
\hline
\hline
Notation & Description \\
\hline
$i, N, \mathbb{N}$ & User, number of users, and the set of users \\
$j, M, \mathbb{S}$ & Task, number of tasks, and the set of tasks \\
$Score(i,j)$ & Recommendation score for user $i$ and task $j$ \\
$c, C$ & Task category, and number of categories \\
$p_{i,j}$ & User $i$'s preference for task $j$ \\
$q_{i,c}$ & User $i$'s reliability in task category $c$ \\
$\mathbb{S}_c$ & The set of tasks in category $c$ \\
$\mathbb{N}_c$ & The set of users contributed data to $\mathbb{S}_c$ \\
$x_{i,j}$ & User $i$'s data for task $j$ \\
$x^*_j$ & Ground truth of task $j$ \\
$\hat{x}^*_j$ & Estimation of the task $j$'s ground truth \\
$y_{i,j}$ & If user $i$ contributed data to task $j$ \\
$\mathbb{S}_{i,c}$ & The set of tasks user $i$ selected in $\mathbb{S}_c$ \\
$\mathbb{D}_{i,c}$ & The set of tasks user $i$ finished in $\mathbb{S}_c$ \\
$r_{i,c}$ & User $i$'s task completion ration in $\mathbb{S}_c$ \\
$\mathbb{O}$ & The set of tasks with known ground truths \\
$\mathbb{O}_c$ & The set of tasks belong to $\mathbb{O}$ and in $\mathbb{S}_c$\\
$\mathbb{N}^o_c$ & The set of users contributed data to $\mathbb{O}_c$ \\
\hline
\hline
\end{tabular}
\end{table}

\section{User Reliability Profiling Algorithm}
\label{algorithm}
In this section, we first propose an algorithm to solve the user reliability profiling problem formulated above. Then, we further propose a matrix factorization method to estimate each user's reliability for the task categories that lack the user's historical performance.

\subsection{Estimating Users' Reliability}

In the problem formulated in Equation \ref{semi-final-simple}, two sets of variables need to be estimated, \ie, the users' reliability levels and unknown ground truths. We propose an efficient block coordinate descent algorithm to solve it. The idea of the algorithm is to fix one set of variables to solve the other, and repeat this process until convergence. Since the estimation process for each category can be done independently, parallel computing can be adopted to speed up the entire calculation process. For each task category $c$, we perform the following three steps: parameter initialization, truth update, and reliability estimation.

\setcounter{subsubsection}{-1}
\subsubsection{\textbf{Parameter Initialization}}

In the parameter initialization phase, we assign initial values to one set of the variables to give the learning algorithm a starting point. Existing truth discovery algorithms either initialize the unknown ground truths using a simple majority voting or averaging scheme, or uniformly initialize the reliability parameters. As pointed out in \cite{li2016survey,yang2017designing}, random or uniform initialization may result in poor estimation performance, which is especially true when most data are unreliable.

To mitigate this problem, we propose to enhance the initialization of the users' reliability parameters $\{q_{i,c}\}$ by incorporating the prior knowledge of available ground truths. The idea is to leverage the known truth knowledge to give related users good initial estimations of their reliability. Specifically, for each category $c$, let $\mathbb{N}_c^o$ denote the set of users who contributed data to tasks in $\mathbb{O}_c$. For the users in $\mathbb{N}_c^o$, we initialize their reliability by solving the following problem.
\begin{align}
\underset{\{q_{i,c}\}, i \in \mathbb{N}_c^o}{\textrm{argmin}} & \,\, \sum_{i \in \mathbb{N}_c^o} \sum_{o \in \mathbb{O}_c} y_{i,o} \, q_{i,c} \, g(r_{i,c}) \, L(x_{i,o}, x^*_o), \notag \\
\textrm{s.t.} & \,\, \sum_{i \in \mathbb{N}_c^o} \textrm{exp}(- q_{i,c} \, g(r_{i,c})) = \frac{|\mathbb{N}_c^o|}{|\mathbb{N}_c|}.
\label{reliability-ini-1}
\end{align}
The above problem is convex, thus we can apply the method of Lagrangian multipliers to solve it.

As for the remaining users in $\mathbb{N}_c \setminus \mathbb{N}_c^o$, since they did not contribute data to tasks whose ground truths are known, no prior knowledge can be applied. Thus, their reliability parameters are uniformly initialized such that
\begin{equation}
\sum_{i \in \mathbb{N}_c \setminus \mathbb{N}_c^o} \textrm{exp}(- q_{i,c} \, g(r_{i,c})) = 1- \frac{|\mathbb{N}_c^o|}{|\mathbb{N}_c|}.
\label{reliability-ini-2}
\end{equation}


Solving Equation \ref{reliability-ini-1} and Equation \ref{reliability-ini-2}, we have the initialization of the users' reliability parameters as follows:
\begin{equation}
q_{i,c} = \left\{
\begin{array}{ll}
\frac{\textrm{log} \Big( \frac{|\mathbb{N}_c| \sum_{i \in \mathbb{N}_c^o} \sum_{o \in \mathbb{O}_c} y_{i,o} L(x_{i,o}, x^*_o)}{|\mathbb{N}_c^o| \sum_{o \in \mathbb{O}_c} y_{i,o} L(x_{i,o}, x^*_o)} \Big)}{g(r_{i,c})} & \textrm{if} \, i \in \mathbb{N}_c^o, \\
\frac{\textrm{log}(|\mathbb{N}_c|)}{g(r_{i,c})} & \textrm{if} \, i \in \mathbb{N}_c \setminus \mathbb{N}_c^o.
\end{array}
\right.
\end{equation}

\subsubsection{\textbf{Truth Update}}

After obtaining an initial estimation of the users' reliability, we can update the estimation of truths by treating the estimated reliability parameters $\{q_{i,c}\}$ as fixed values. Then, the truth of each task $j \in \mathbb{S}_c$ can be updated using the following rule.
\begin{equation}
\{\hat{x}_j^*\} \leftarrow \underset{\{\hat{x}_j^*\},j \in \mathbb{S}_c }{\textrm{argmin}} \sum_{i \in \mathbb{N}_c} q_{i,c} \, g(r_{i,c}) \Big( \sum_{j \in \mathbb{S}_c} y_{i,j} \, L(x_{i,j}, \hat{x}^*_j) + \alpha \epsilon_{i,c} \Big)
\label{truth-update}
\end{equation}

\begin{theorem} 
Given the users' reliability parameters, the optimization problem in Equation \ref{truth-update} can be optimally solved. For continuous data type, the optimal solution is given by
\begin{equation}
\hat{x}_j^* = \frac{\sum_{i \in \mathbb{N}_c} q_{i,c} \, y_{i,j} \, x_{i,j} \, g(r_{i,c})}{\sum_{i \in \mathbb{N}_c} q_{i,c} \, y_{i,j} \, g(r_{i,c})}.
\label{hahaah1}
\end{equation}
As for categorial data type, the solution is
\begin{equation}
\hat{x}_j^* = \underset{x'_j \in \{x_{i,j}\}}{\textrm{argmax}} \sum_{i \in \mathbb{N}_c} q_{i,c} \, y_{i,j} \, g(r_{i,c}) \, \bm{1}(x_{i,j}, x'_j),
\label{hahaha2}
\end{equation}
where $\bm{1}(x,y)=1$ if $x=y$, and 0 otherwise.
\end{theorem}

\begin{proof}
For each task $j$, we first consider the case of continuous data, where $L(x_{i,j}, \hat{x}_j^*)=(x_{i,j} - \hat{x}_j^*)^2$. Then, the objective function can be formalized as follows
$$
f(\{\hat{x}_j^*\}) = \sum_{i \in \mathbb{N}_c} q_{i,c} \, g(r_{i,c}) \Big(  \sum_{j \in \mathbb{S}_c} y_{i,j} \, (x_{i,j} - \hat{x}_j^*)^2 + \alpha \epsilon_{i,c} \Big).
$$
We take the partial derivative of the function with respect to $\hat{x}_j^*$ and set it to zero, \ie,
$$
\frac{\partial f}{\partial \hat{x}_j^*} = 2 \sum_{i \in \mathbb{N}_c} q_{i,c} \, g(r_{i,c}) \, y_{i,j} (\hat{x}_j^* - x_{i,j})=0
$$
Solving the above equation, we get Equation \ref{hahaah1}.

For some task $j$, if its data $x_{i,j}$ is of categorical type, the loss function is 
\begin{equation}
L(x_{i,j}, \hat{x}_j^*) = \left\{
\begin{array}{ll}
0 & \textrm{if} \, x_{i,j} = \hat{x}_j^*, \\
1 & \textrm{otherwise}.
\end{array}
\right.
\end{equation}
Taking the loss function into Equation \ref{truth-update}, we can get the optimal solution to $\hat{x}_j^*$ shown in Equation \ref{hahaha2}.
\end{proof}

\subsubsection{\textbf{Reliability Estimation}}

After updating the estimation of the ground truth, we now fix the values of $\{\hat{x}_j^*\}$, and calculate the users' data qualities $\{q_{i,c}\}$ by solving the following optimization function. Intuitively, the users whose data are close to the ground truth will have high reliability estimations, and vice versa.
\begin{align}
\{q_{i,c}\} \leftarrow \underset{\{q_{i,c}\}}{\textrm{argmin}} \, & \sum_{i \in \mathbb{N}_c} q_{i,c} \, g(r_{i,c}) \Big( \sum_{j \in \mathbb{S}_c} y_{i,j} \, L(x_{i,j}, \hat{x}^*_j) + \alpha \epsilon_{i,c} \Big) \notag \\
\textrm{s.t.} \,\,\, & \sum_{i \in \mathbb{N}_c} \textrm{exp}(- q_{i,c} \, g(r_{i,c})) = 1.
\label{quality-update}
\end{align}

\begin{theorem}
Given fixed truth estimation $\{\hat{x}_j^*\}$, the problem in Equation \ref{quality-update} can be optimally solved. The optimal value of each $q_{i,c}, i \in \mathbb{N}_c$ is given by
\begin{equation}
q_{i,c} = \frac{1}{g(r_{i,c})} \textrm{log} \Big( \frac{\sum_{i \in \mathbb{N}_c} \big( \sum_{j \in \mathbb{S}_c} y_{i,j} L(x_{i,j}, \hat{x}^*_j) + \alpha \epsilon_{i,c} \big) }{\sum_{j \in \mathbb{S}_c} y_{i,j} L(x_{i,j}, \hat{x}^*_j) + \alpha \epsilon_{i,c} } \Big).
\label{theorem-2}
\end{equation}
\end{theorem}

\begin{proof}
The problem is convex, since the objective term is linear and the constraint set is convex. Therefore, we can apply the method of Lagrangian multipliers to solve the problem. The Lagrangian of Equation \ref{quality-update} is given as:
\begin{align}
f(\{q_{i,c}\}, \lambda) = & \sum_{i \in \mathbb{N}_c} q_{i,c} \, g(r_{i,c}) \Big( \sum_{j \in \mathbb{S}_c} y_{i,j} \, L(x_{i,j}, \hat{x}^*_j) + \alpha \epsilon_{i,c} \Big) \notag \\
& + \lambda \Big(\sum_{i \in \mathbb{N}_c} \textrm{exp}(- q_{i,c} \, g(r_{i,c})) - 1 \Big),
\label{category-truth}
\end{align}
where $\lambda$ is a Lagrange multiplier. Taking the partial derivative of Equation \ref{category-truth} with respect to $q_{i,c}$, we have
\begin{align}
\frac{\partial f}{\partial q_{i,c}} = \, & g(r_{i,c}) \Big( \sum_{j \in \mathbb{S}_c} y_{i,j} L(x_{i,j}, \hat{x}^*_j) + \alpha \epsilon_{i,c} \Big) \notag \\
\, & - \lambda \, g(r_{i,c}) \, \textrm{exp}(- q_{i,c} \, g(r_{i,c})).
\label{partial-der}
\end{align}

Letting Equation \ref{partial-der} to zero, we get
\begin{equation}
\sum_{j \in \mathbb{S}_c} y_{i,j} L(x_{i,j}, \hat{x}^*_j) + \alpha \epsilon_{i,c} = \lambda \, \textrm{exp}(- q_{i,c} \, g(r_{i,c})).
\label{lambda-quality}
\end{equation}

Summing both sides over $i$, we get
\begin{equation}
\sum_{i \in \mathbb{N}_c} \Big( \sum_{j \in \mathbb{S}_c} y_{i,j} L(x_{i,j}, \hat{x}^*_j) + \alpha \epsilon_{i,c} \Big) = \lambda \sum_{i \in \mathbb{N}_c} \textrm{exp}(- q_{i,c} \, g(r_{i,c})).
\end{equation}

Since $\sum_{i \in \mathbb{N}_c} \, \textrm{exp}(- q_{i,c} \, g(r_{i,c})) = 1$, we have
\begin{equation}
\lambda = \sum_{i \in \mathbb{N}_c} \Big( \sum_{j \in \mathbb{S}_c} y_{i,j} L(x_{i,j}, \hat{x}^*_j) + \alpha \epsilon_{i,c} \Big).
\label{get-lambda}
\end{equation}

Taking Equation \ref{get-lambda} into Equation \ref{lambda-quality}, we obtain a closed form solution of reliability $q_{i,c}$ shown in Equation \ref{theorem-2}.
\end{proof}

\begin{algorithm}[t] \small
\caption{User Reliability Estimation (Category $c$)}
\KwIn{Tasks $\mathbb{S}_c$ and $\mathbb{O}_c$, users $\mathbb{N}_c$, and data $\{x_{i,j}\}$}
\KwOut{Reliability $\{q_{i,c}\}$, and truth estimation $\{\hat{x}_j^*\}$}
\tcp{\small Parameter Initialization:}
\If{$i \in \mathbb{N}_c$}{
    \If{$i \in \mathbb{N}_c^o$}{
        $q_{i,c} \leftarrow \frac{1}{g(r_{i,c})} \textrm{log} \Big( \frac{|\mathbb{N}_c| \sum_{i \in \mathbb{N}_c^o} \sum_{o \in \mathbb{O}_c} y_{i,o} L(x_{i,o}, x^*_o)}{|\mathbb{N}_c^o| \sum_{o \in \mathbb{O}_c} y_{i,o} L(x_{i,o}, x^*_o)} \Big)$\;
    }
    \lElse{
        $q_{i,c} \leftarrow \frac{\textrm{log}(|\mathbb{N}_c|)}{g(r_{i,c})}$\;
    }
}
\lElse{
    $q_{i,c} \leftarrow \textrm{N/A}$\;
}
\While{not converged}{
    \tcp{\small Truth Update}
    \ForEach{task $j \in \mathbb{S}_c$}{
        \If{the task $j$ is of continuous data type}{
            $\hat{x}_j^* \leftarrow \frac{\sum_{i \in \mathbb{N}_c} q_{i,c} \, y_{i,j} \, x_{i,j} \, g(r_{i,c}) }{\sum_{i \in \mathbb{N}_c} q_{i,c} \, y_{i,j} \, g(r_{i,c}) }$\;
        }
        \If{the task $j$ is of categorical data type}{
            $\hat{x}_j^* \leftarrow \underset{x'_j \in \{x_{i,j}\}}{\textrm{argmax}} \sum_{i \in \mathbb{N}_c} q_{i,c} \, y_{i,j} \, g(r_{i,c}) \bm{1}(x_{i,j}, x'_j)$\;
        }
    }
    \tcp{\small Reliability Estimation}
    \ForEach{user $i \in \mathbb{N}_c$}{
        $q_{i,c} \leftarrow \frac{1}{g(r_{i,c})} \textrm{log} \Big( \frac{\sum_{i \in \mathbb{N}_c} \big( \sum_{j \in \mathbb{S}_c} y_{i,j} L(x_{i,j}, \hat{x}^*_j) + \alpha \epsilon_{i,c} \big) }{ \sum_{j \in \mathbb{S}_c} y_{i,j} L(x_{i,j}, \hat{x}^*_j) + \alpha \epsilon_{i,c}} \Big)$\;
    }
}
\tcp{Reliability Normalization}
$\forall i \in \mathbb{N}_c, q_{i,c} \leftarrow \frac{q_{i,c}}{\textrm{log}|\mathbb{N}_c|}$\;
\Return{\{$q_{i,c}$\} \textrm{and} $\{\hat{x}^*_j\}$}
\end{algorithm}

The pseudo-code of the algorithm is presented in Algorithm 1. We first initialize the users' reliability parameters, and then keep iterating the steps of truth update and reliability estimation until the change of the users' reliability is below a certain threshold. Due to the convexity of our problem and the ability to achieve the optimal solution for each step (Theorem 1 and Theorem 2), our algorithm is guaranteed to converge to some local optimum, according to the proposition of the block coordinate descent \cite{bertsekas1999nonlinear}. Further improvements can be made to find a 2-approximation of the global optimum within nearly linear time \cite{ding2016finding}.

\subsubsection{\textbf{Reliability Normalization}}

Until now, we have obtained the estimations of the users' reliability levels and unknown ground truths. However, there is a problem in our model, \ie, each user $i$'s reliability estimations for different categories are in different scales. From the regularization term $\sum_{i \in \mathbb{N}_c} \textrm{exp}(-q_{i,c} g(r_{i,c})) = 1$, we can see that the average value for $q_{i,c} g(r_{i,c})$ is $\textrm{log}|\mathbb{N}_c|$, which is proportional to the number of data contributors for category $c$. This means that a user is likely to receive a higher reliability score when she is among a large number of data contributors, which is not reasonable. In order to guarantee each user's reliability estimations for different tasks are in the same scale, we normalize each user $i$'s reliability estimation $q_{i,c}$ into $\frac{q_{i,c}}{\textrm{log}|\mathbb{N}_c|}$.

\subsection{Estimating Missing Entries: A Latent Factor Model}
\label{subsection-missing}

From the above subsection, we have obtained each user's reliability information over the task categories that she has contributed data to. However, we observe that if a user $i$ did not contribute data to some category $c$ (\ie, $i \notin \mathbb{N}_c$), then Algorithm 1 is not able to estimate the user $i$'s reliability over $c$. In this part, we propose a matrix factorization method to address this problem.

We use $Q$ to denote the users' reliability matrix, where each entry $q_{i,c}$ is the user $i$'s reliability for task category $c$. We map both users and task categories to a joint latent factor space of dimensionality $k$. Specifically, we assume that each user $i$ is associated with a vector $\bm{w}_i \in \mathbb{R}^k$, and each category is associated with $\bm{\theta}_c \in \mathbb{R}^k$. The vector $\bm{w}_i=[w_{i,1}, w_{i,2}, \ldots, w_{i,k}]^T$ can be interpreted as the user $i$'s capabilities in $k$ different dimensions, and the vector $\bm{\theta}_c =[\theta_{c,1},\theta_{c,2},\ldots, \theta_{c,k}]^T$ can be seen as the weight of each capability needed by the category $c$. Then, each user $i$'s reliability for each category $c$ can be calculated as $q_{i,c} = \bm {w}_i^T \bm{\theta}_c$.

To estimate the missing entries in matrix $Q$, we tend to calculate each user $i$'s latent vector $\bm{w}_i$ and each category's latent vector $\bm{\theta}_c$. Let $\mathcal{W}$ and $\Theta$ denote the sets of users' and categories' latent vectors, respectively. Then, the objective function can be formalized as follows.
\begin{equation}
\underset{\mathcal{W},\Theta}{\textrm{min}} \,\, \frac{1}{2} \sum_{c=1}^C \sum_{i=1}^N z_{i,c} (q_{i,c} - \bm{w}_i^T \bm{\theta}_c)^2
\label{profile}
\end{equation}
where $z_{i,c}$ indicates if user $i$ has contributed data to category $c$ (1 means yes, and 0 otherwise). To prevent over-fitting, we add regularization terms in Equation \ref{profile}.
\begin{equation}\small
\underset{\mathcal{W},\Theta}{\textrm{min}} \,\,
\frac{1}{2} \sum_{c=1}^C \sum_{i=1}^N a_{i,c} (q_{i,c} - \bm{w}_i^T \bm{\theta}_c)^2 + \frac{\lambda_1}{2} \sum_{i=1}^N \Vert \bm{w}_i \Vert^2 + \frac{\lambda_2}{2} \sum_{c=1}^C \Vert \bm{\theta}_c \Vert^2,
\label{profile-reg}
\end{equation}
where $\Vert \bm{w}_i \Vert^2 = \sum_{t=1}^k w_{i,t}^2$ and $\Vert \bm{\theta}_c \Vert^2 = \sum_{t=1}^k \theta_{c,t}^2$. $\lambda_1$ and $\lambda_2$ are parameters controlling the weights of regularization terms.

\begin{algorithm}[t] \small
\caption{Unknown Reliability Estimation}
\KwIn{Users reliability matrix $Q$}
\KwOut{Unknown reliability parameters $\{q_{i,c}| z_{i,c}=0\}$}
Initialize $\{\bm{w}_i\}$ and $\{\bm{\theta}_c\}$ to small random values\;
\While{not converged}{
    \ForEach{i=1,\ldots,N, c=1,\ldots,C}{
        $w_{i,t} \leftarrow w_{i,t} - \beta \Big( \sum_{c=1}^C z_{i,c} (q_{i,c} - \bm{w}_i^T \bm{\theta}_c) + \lambda_1 w_{i,t} \Big)$,
        $\theta_{c,t} \leftarrow \theta_{c,t} - \beta \Big( \sum_{i=1}^N z_{i,c} (q_{i,c} - \bm{w}_i^T \bm{\theta}_c) + \lambda_2 \theta_{c,t} \Big)$\;
    }
}
\ForEach{$q_{i,c} = \textrm{N/A}$}{
    $q_{i,c} \leftarrow \bm{w}^T_i \cdot \bm{\theta}_c$
}
\Return{$\{q_{i,c}| z_{i,c}=0\}$}
\end{algorithm}

We propose to use a simple gradient descent method to solve the above problem. The pseudo-code is presented in Algorithm 2. We first initialize $\{w_{i,t}\}$ and $\{\theta_{c,t}\}$ to small random values. After that, we apply gradient descent algorithm, \ie, for every $i$ and $t$, we update $\{w_{i,t}\}$ and $\{\theta_{c,t}\}$ using the following rules
\begin{align}
& w_{i,t} \leftarrow w_{i,t} - \beta \Big( \sum_{c=1}^C z_{i,c} (q_{i,c} - \bm{w}_i^T \bm{\theta}_c) + \lambda_1 w_{i,t} \Big), \\
& \theta_{c,t} \leftarrow \theta_{c,t} - \beta \Big( \sum_{i=1}^N z_{i,c} (q_{i,c} - \bm{w}_i^T \bm{\theta}_c) + \lambda_2 \theta_{c,t} \Big),
\end{align}
where $\beta$ is the learning rate. Finally, we can predict a user $i$'s reliability for a task category $c$ even if the user $i$ did not provide any data to $c$, \ie, for $i \notin \mathbb{N}_c$, $q_{i,c} \leftarrow \bm{w}^T_i \bm{\theta}_c$.

\section{Evaluation}
\label{evaluation}
In this section, we implement and evaluate the performance of our proposed methods. We first conduct a real-world crowdsensing experiment, and then simulate a large-scale crowdsensing scenario to further examine the performance of our methods.

\subsection{Experiment Setup}

We recruit 10 users (8 males and 2 females) to participate in our experiment. In the experiment, we manually create 123 sensing tasks for 9 different categories. An overview of the tasks is presented in Table \ref{task-overview}. The tasks within the same category focus on the same sensing target (such as noise, traffic, or weather), but with different attributes, including time, locations, and payments. Each task category has a data type requirement. For instance, noise monitoring requires continuous data type, while weather monitoring requires categorical data type. The entire task corpus is shown to the users through the browsers on the users' smartphones. Each user can browse through these tasks, and choose their interested tasks to work on. The ground truth of each task is monitored by the authors themselves, and unavailable to the users. We collect the users' sensing data, as well as their operation records, including each user's task browsing history, task selection history, and task completion history.

According to our collected data, each user contributes data to about 60\% of the tasks in average. The parameter $\alpha$ used in our semi-supervised learning model is set to 1. And for each task category, we use the ground truths of 10\% of the tasks. The parameters $k$, $\lambda1$ and $\lambda_2$ used in our matrix factorization method are set to 3, 5 and 5, respectively.


\begin{table}[!t] \centering \small
\caption{Sensing Task Overview}
\begin{tabular}{llll}
\hline
Category & Monitored target & \# of tasks & Data type \\
\hline
C1 & noise & 8 & continuous \\
C2 & air pollution & 9 & categorical \\
C3 & traffic congestion & 11 & categorical \\
C4 & human flow & 20 & categorical \\
C5 & temperature & 20 & continuous \\
C6 & weather & 9 & categorical \\
C7 & price & 12 & continuous \\
C8 & question & 17 & categorical \\
C9 & accident & 17 & categorical \\
\hline
\end{tabular}
\label{task-overview}
\vspace{-0.2cm}
\end{table}

\subsection{Experiment Results on User Reliability Profiling}

\begin{figure}[!t]
\centering
\subfigure[Continuous Data]
{\includegraphics[width=0.24\textwidth]{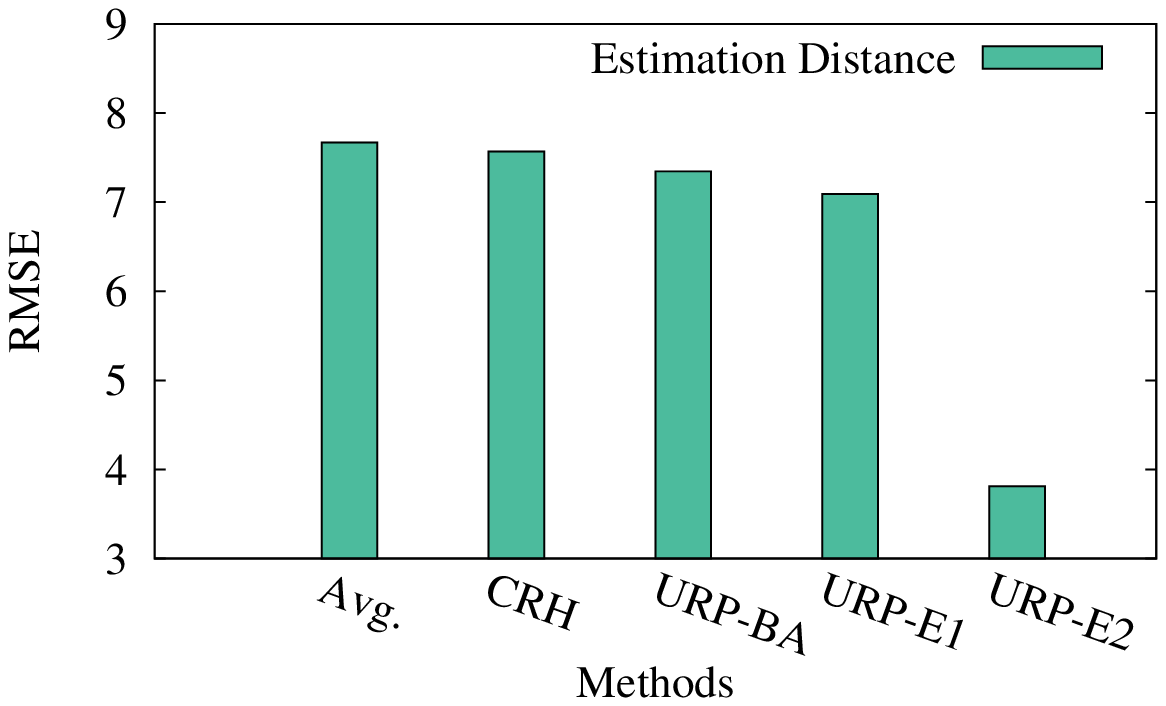}}
\subfigure[Categorical Data]
{\includegraphics[width=0.24\textwidth]{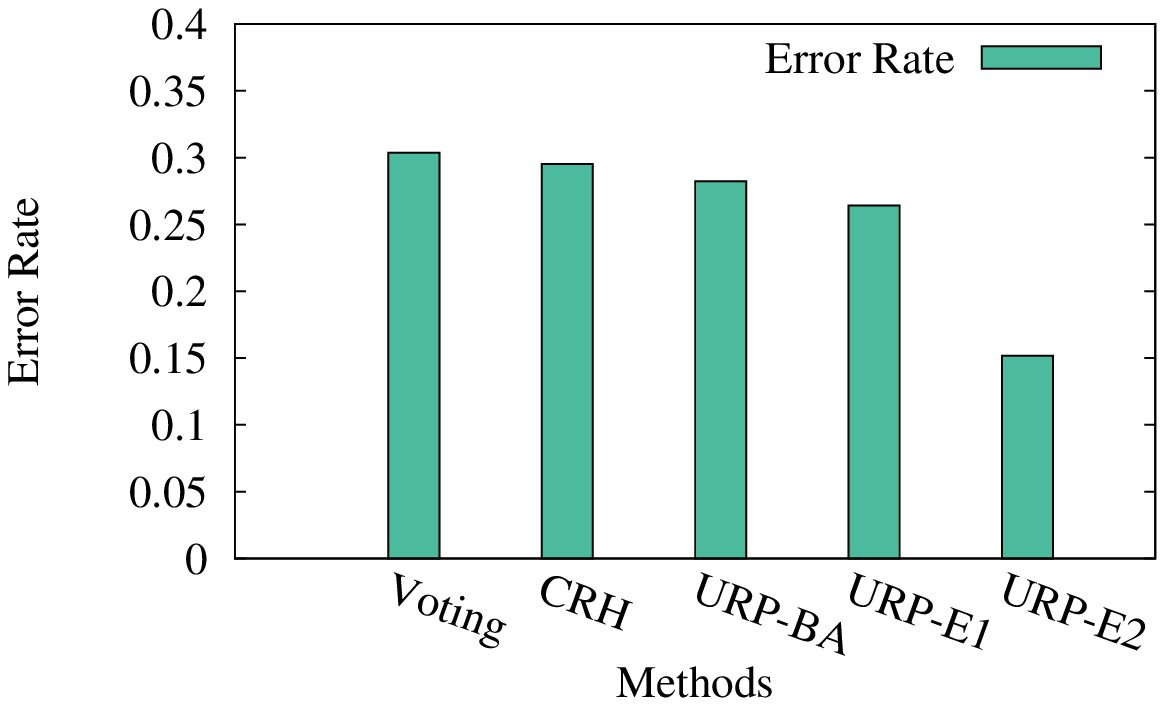}}
\vspace{-0.4cm}
\caption{Performance Comparison on Estimation Accuracay}
\label{experiment-accuracy}
\vspace{-0.4cm}
\end{figure}

\begin{figure}[t]
\centering
\subfigure[Reliability]
{\includegraphics[width=0.18\textwidth]{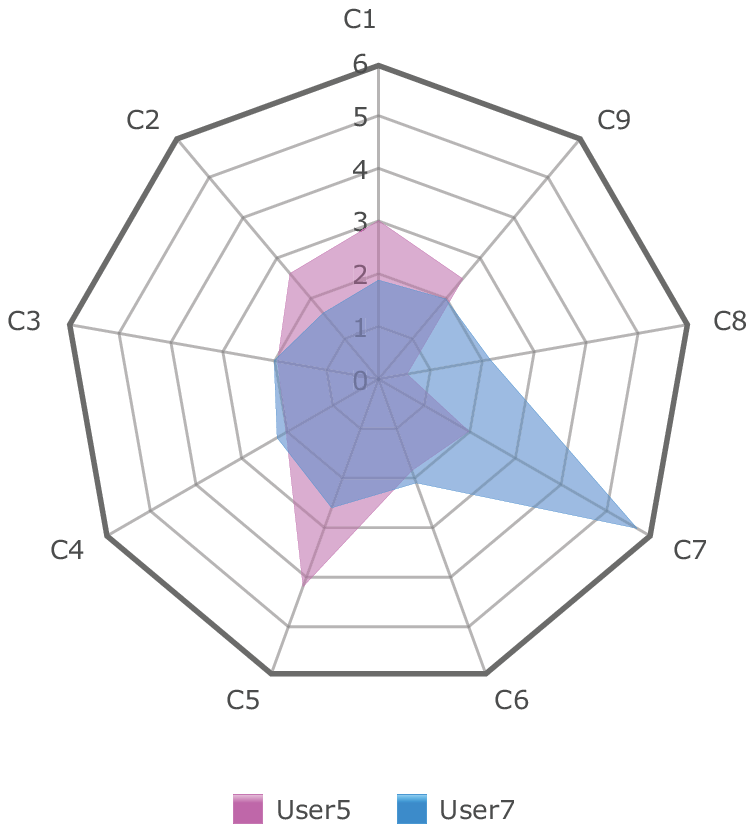}}
\hspace{0.2cm}
\subfigure[Preference]
{\includegraphics[width=0.18\textwidth]{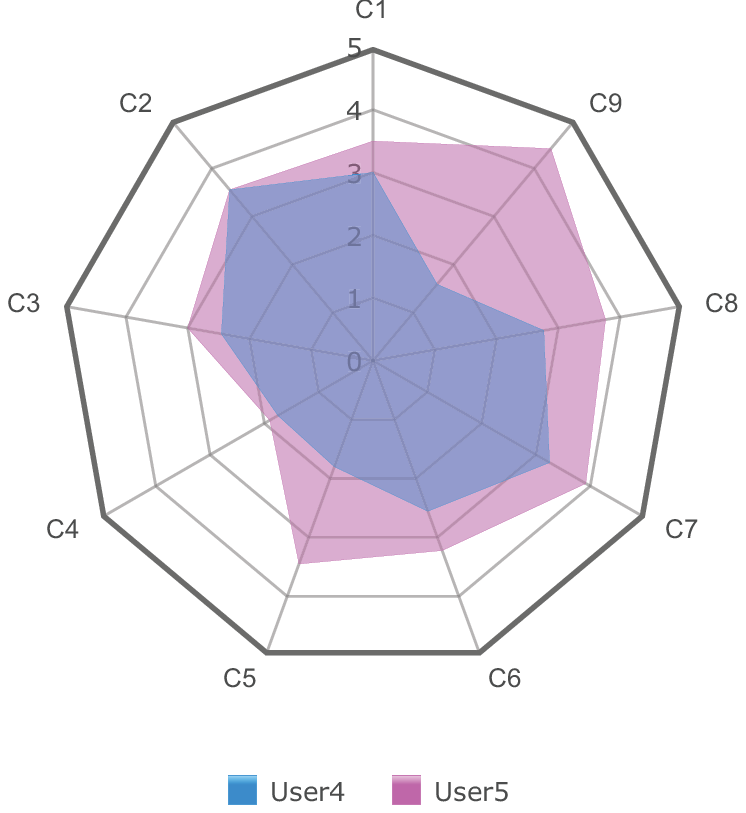}}
\vspace{-0.2cm}
\caption{User Profiling}
\label{user-profiling}
\vspace{-0.5cm}
\end{figure}

In the experiment, we evaluate the performance of our proposed user profiling algorithm. To differentiate the notations, we use ``URP-BA'' to denote the basic version shown in \ref{urp-basic}, and ``URP-E1'' and ``URP-E2'' to denote the first enhancement and the second enhancement, respectively. We compare our algorithms with two benchmarks. One is a heuristic method that treats each user's data equally, \ie, simple average (``Avg.'') for continuous data and majority voting (``Voting'') for categorical data. The other benchmark is a general truth discovery framework, called ``CRH'' \cite{li2014resolving}, which uses a single parameter to model each user's reliability level. We adopt the following two metrics to measure the performance of the algorithms.
\begin{itemize}
\item RMSE: For continuous data, we use Root Mean Square Error (RMSE) to measure the distance between the estimation result and the ground truth. Mathematically, the RMSE is defined as $\sqrt{\sum_{j \in \mathbb{S}} (x_j^* - \hat{x}^*_j)^2/|M|}$.
\item Error Rate: For categorical data, we use Error Rate to quantify the performance of an algorithm. The Error Rate of an algorithm is defined as the percentage of the tasks to which the algorithm's estimations are different from the ground truth, \ie, $1 - \frac{\sum_{j \in \mathbb{S}} \bm{1}(x_j^*,\hat{x}^*_j)}{M}$.
\end{itemize}

Fig. \ref{experiment-accuracy} presents the performance comparison between our algorithms and the benchmarks. We can see that for either data type, the truth discovery-based algorithms can achieve higher estimation accuracy than the simple average or majority voting, indicating the effectiveness of truth discovery algorithms. However, the performance of Avg./Voting, CRH, URP-BA, and URP-E1 tends to be similar. The main reason is that under the crowdsensing scenarios, these usually exist many tasks to which the majority of the users' data are inaccurate, thus the traditional unsupervised learning models may have trouble identifying the users' true reliability levels. In this case, as we can see that URP-E2 has superior performance to the other four algorithms, incorporating even a small number of ground truths can greatly improve the estimation accuracy.

\subsection{Experiment Results on Personalized Task Matching}

Besides profiling the users' reliability, we also profile each user's preference towards each task using the methods proposed in Section \ref{sec-pre-profiling}. In Fig. \ref{user-profiling}(a) and Fig. \ref{user-profiling}(b), we present the reliability profiles and preference profiles of two representative users respectively, where the user's preference towards a task category is calculated as the user's average preference score of the tasks in the category. We normalize the users' preferences to [0,5] for better graphical presentation.

To evaluate the performance of our personalized task recommender system, we provide each user a list of 20 recommended tasks, and ask each user to choose their interested tasks. Recall that our personalized task recommender system recommends tasks to the users based on both the users' reliability and preference. Specifically, for each user and task pair $(i,j)$, we calculate a recommendation score $Score(i,j) = \gamma p_{i,j} + (1-\gamma) q_{i,j}$. Suppose task $j$ belongs to category $c$, then we set $p_{i,j}$ to $p_{i,c}$. We use $\gamma = 0.4$ and $\eta = 0.5$ in our experiment. After that, our system recommends each user 20 tasks with the highest recommendation scores. Three benchmarks are adopted, including random recommendation, preference-only recommendation, and reliability-only recommendation. Random task recommendation strategy provides each user a list of 20 randomly chosen tasks, while the preference- or reliability-only recommendation strategies provide each user 20 tasks with highest preference or reliability scores, respectively.

\begin{figure}[t]
\centering
\subfigure[Acceptance Ratio]
{\includegraphics[width=0.24\textwidth]{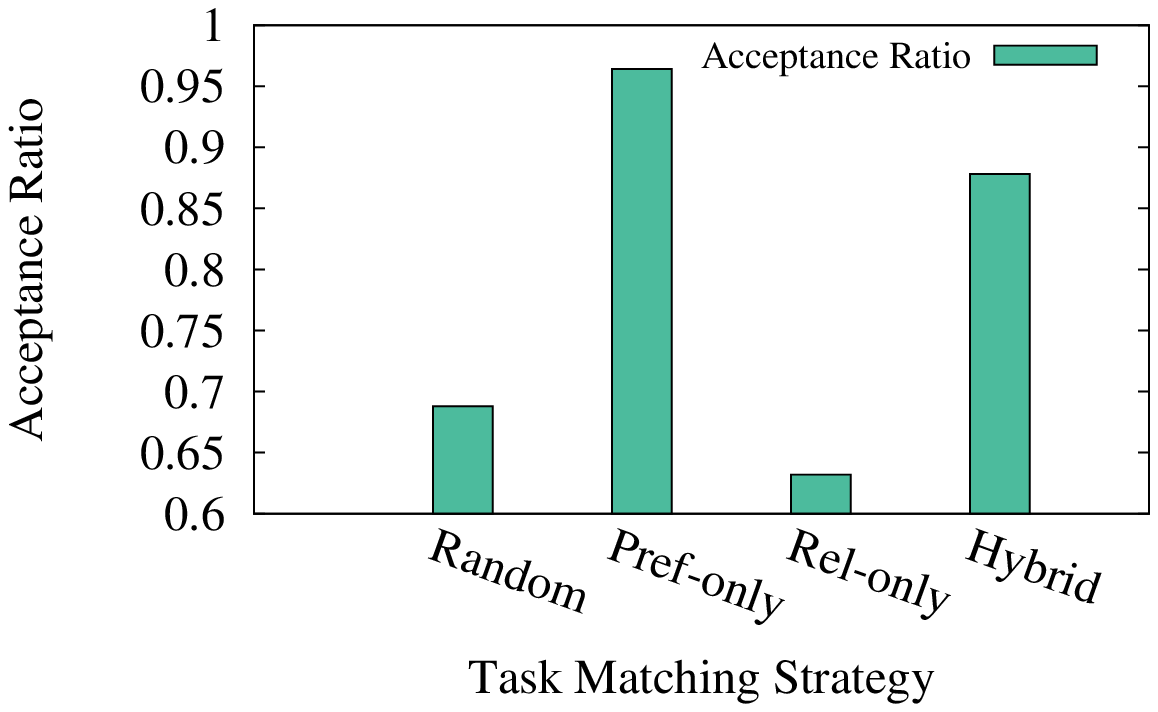}}
\subfigure[Estimation Accuracy]
{\includegraphics[width=0.24\textwidth]{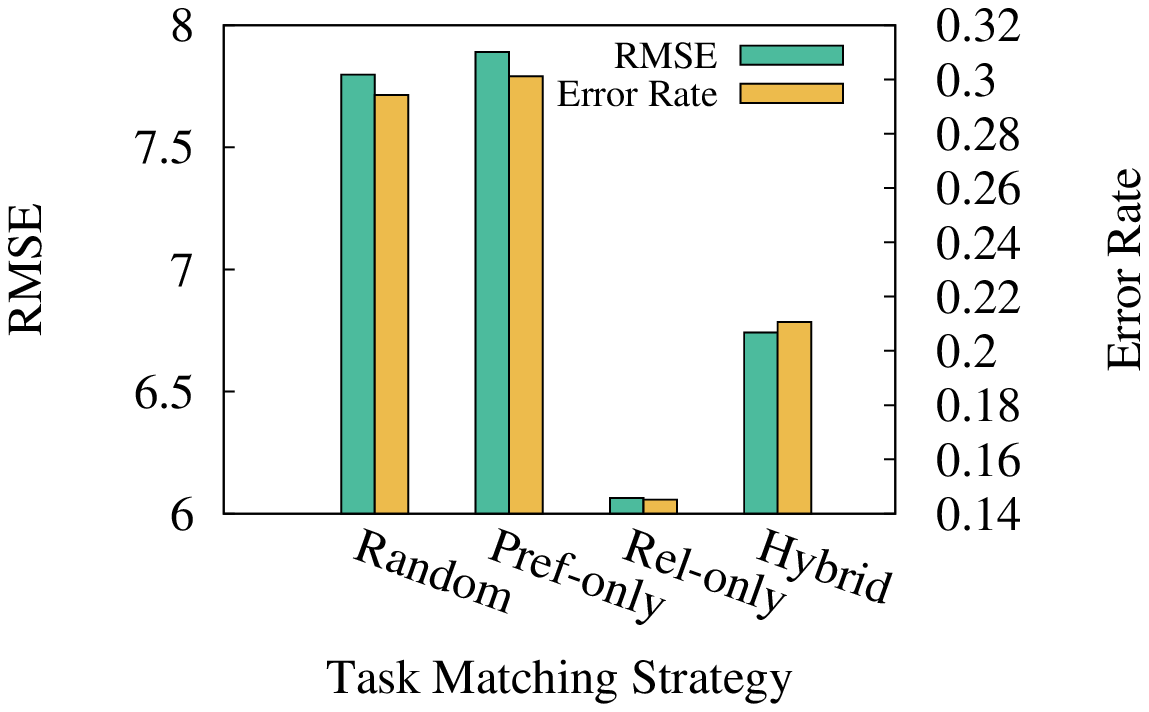}}
\vspace{-0.3cm}
\caption{Comparison on Different Task Matching Strategies}
\label{task-matching-cmp}
\vspace{-0.4cm}
\end{figure}

%

\begin{figure*}[!ht]
\centering
\subfigure[Setting1]
{\includegraphics[width=0.3\textwidth]{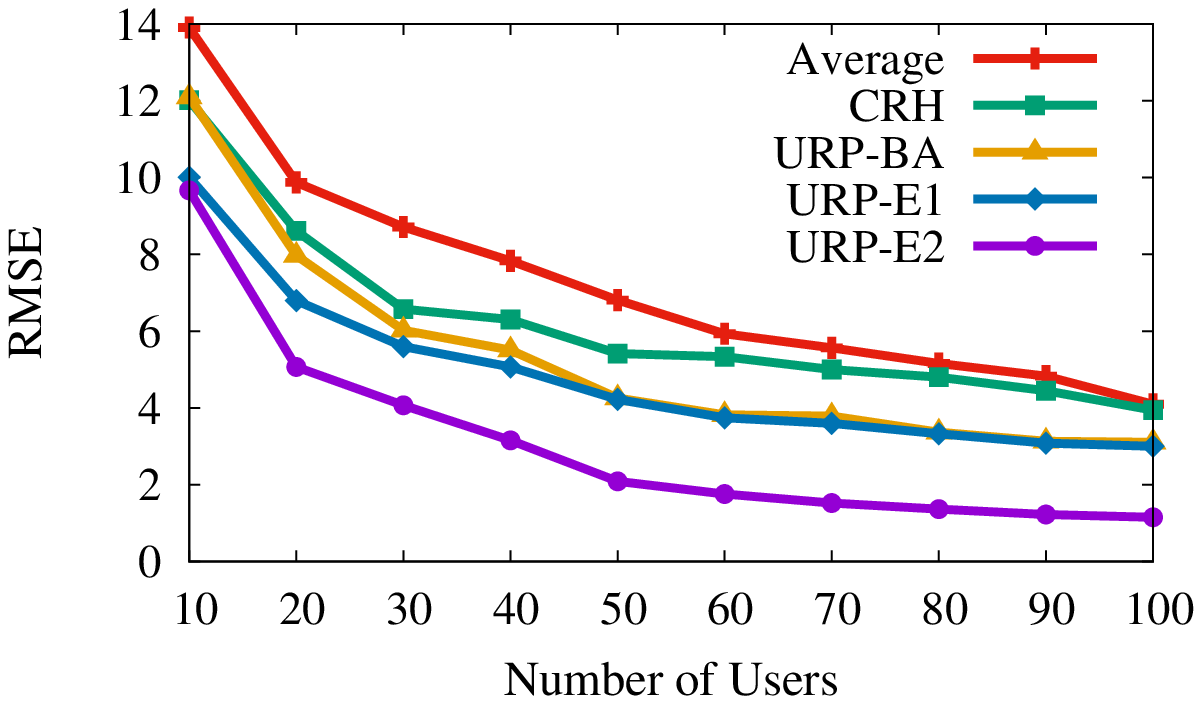}}
\subfigure[Setting2]
{\includegraphics[width=0.3\textwidth]{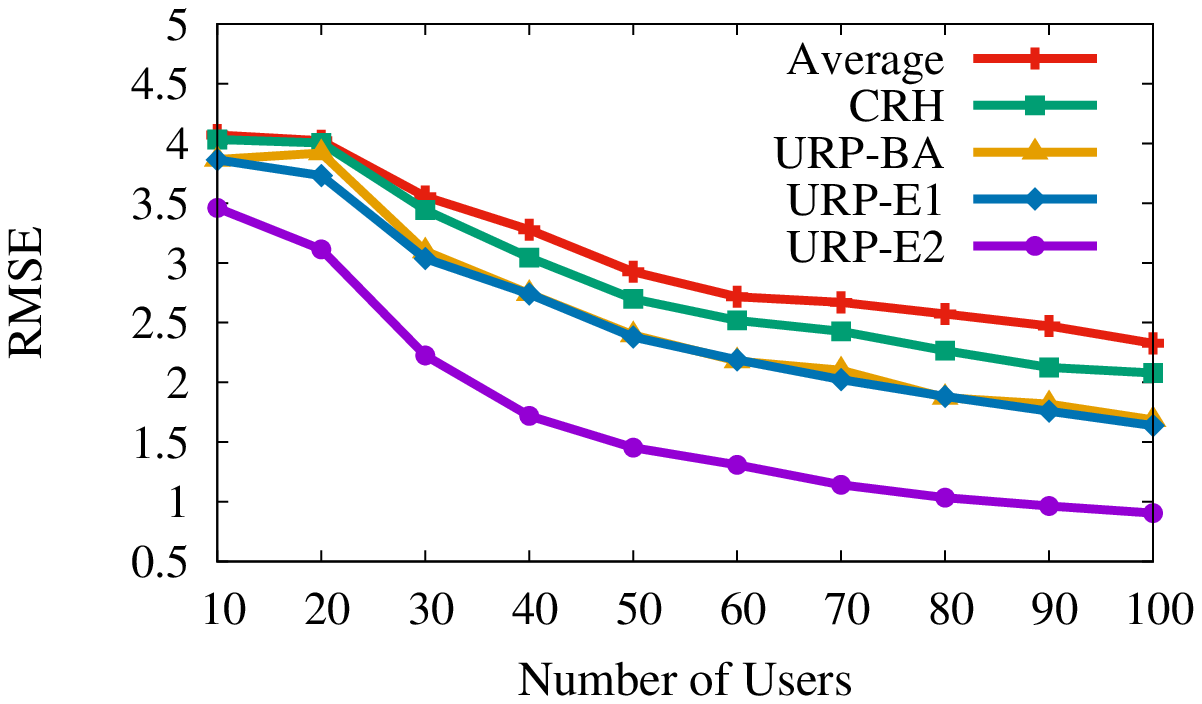}}
\subfigure[Setting3]
{\includegraphics[width=0.3\textwidth]{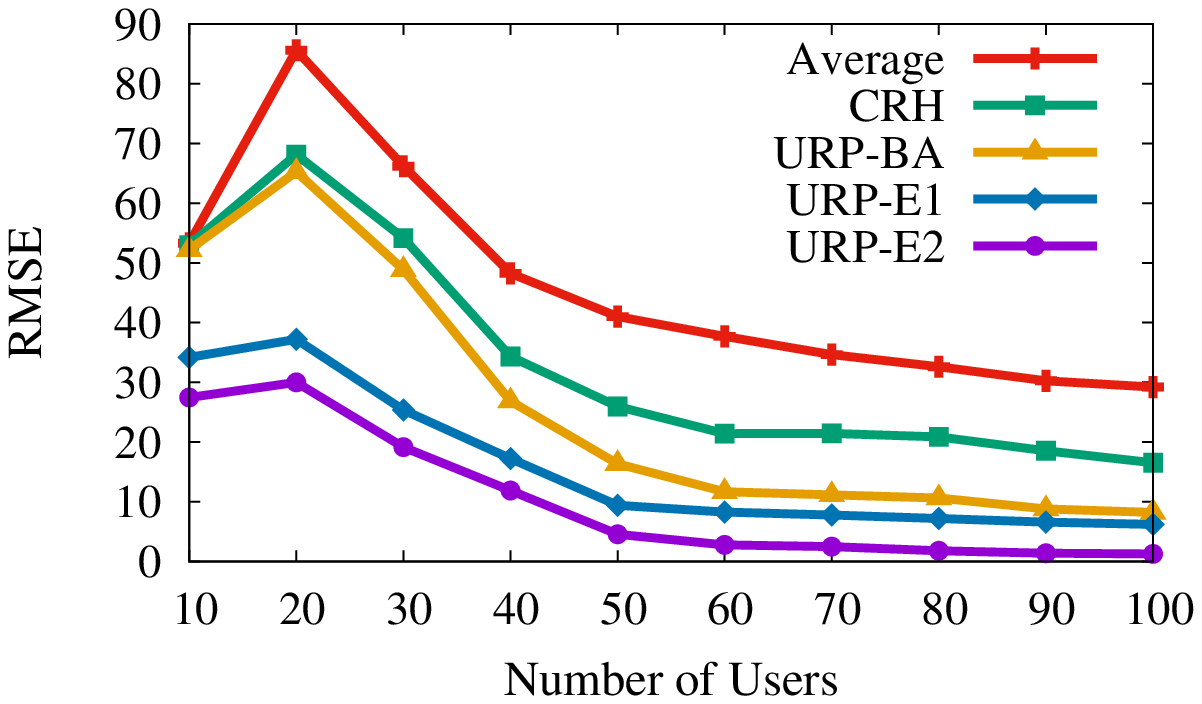}}
\vspace{-0.3cm}
\caption{Comparisons on estimation accuracy with varying number of users}
\label{varying-num-users}
\vspace{-0.5cm}
\end{figure*}

\begin{figure*}[!ht]
\centering
\subfigure[Setting1]
{\includegraphics[width=0.3\textwidth]{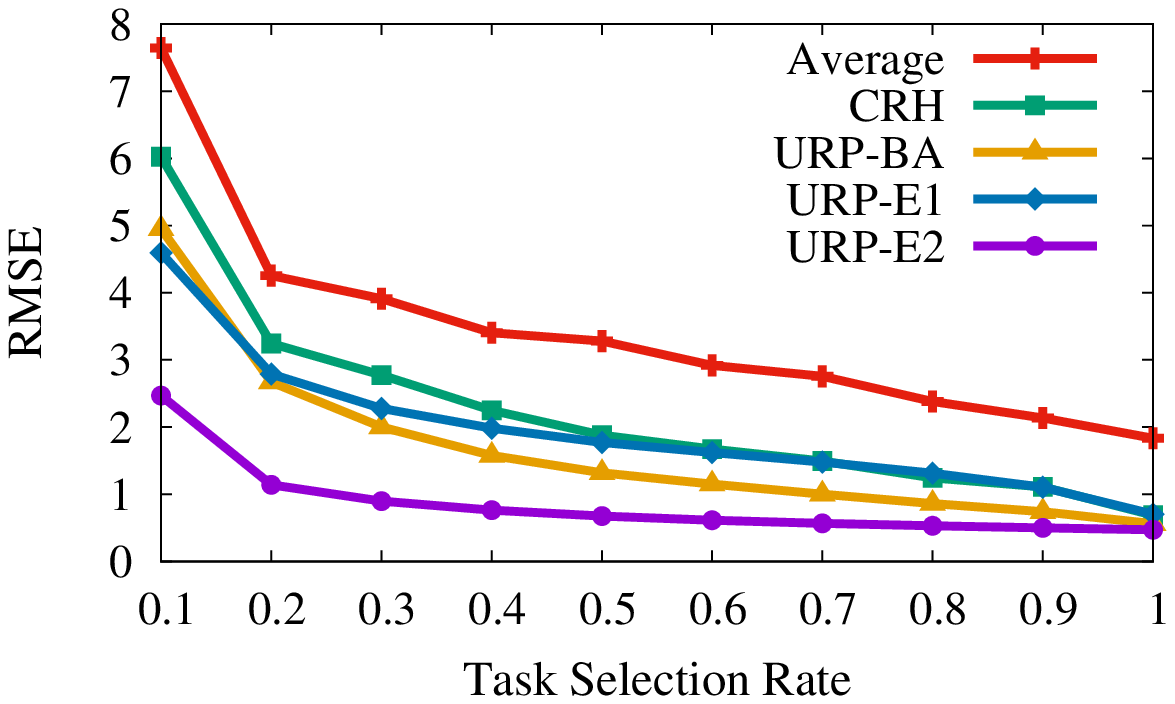}}
\subfigure[Setting2]
{\includegraphics[width=0.3\textwidth]{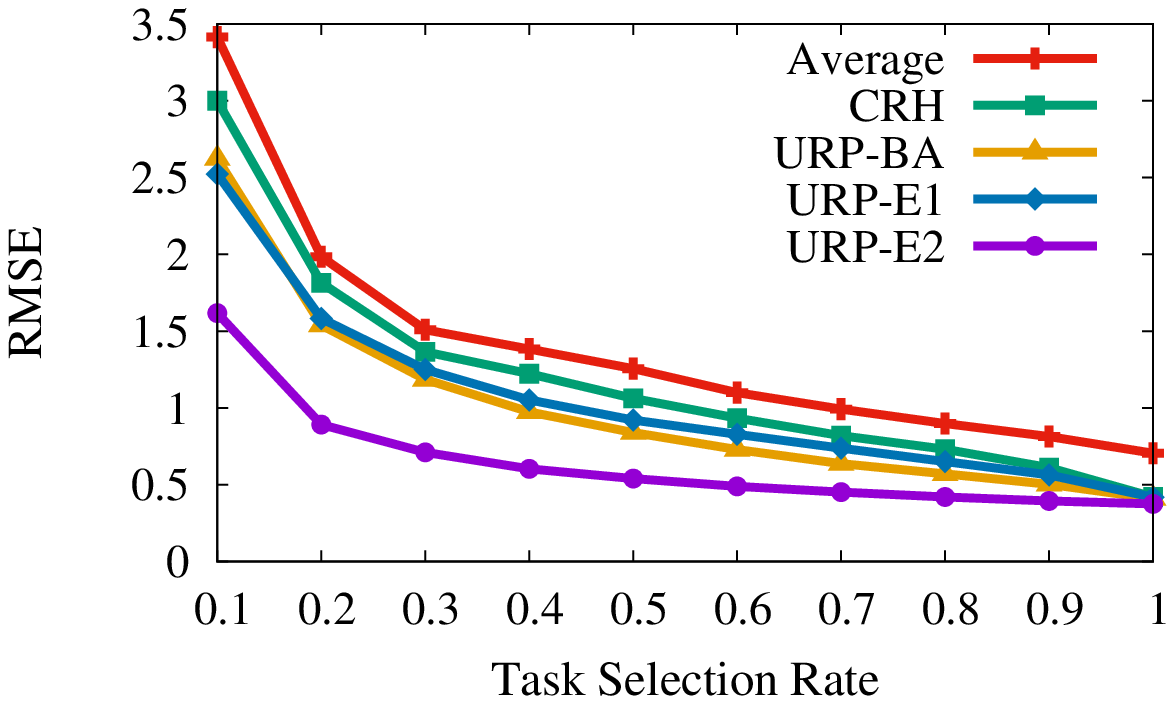}}
\subfigure[Setting3]
{\includegraphics[width=0.3\textwidth]{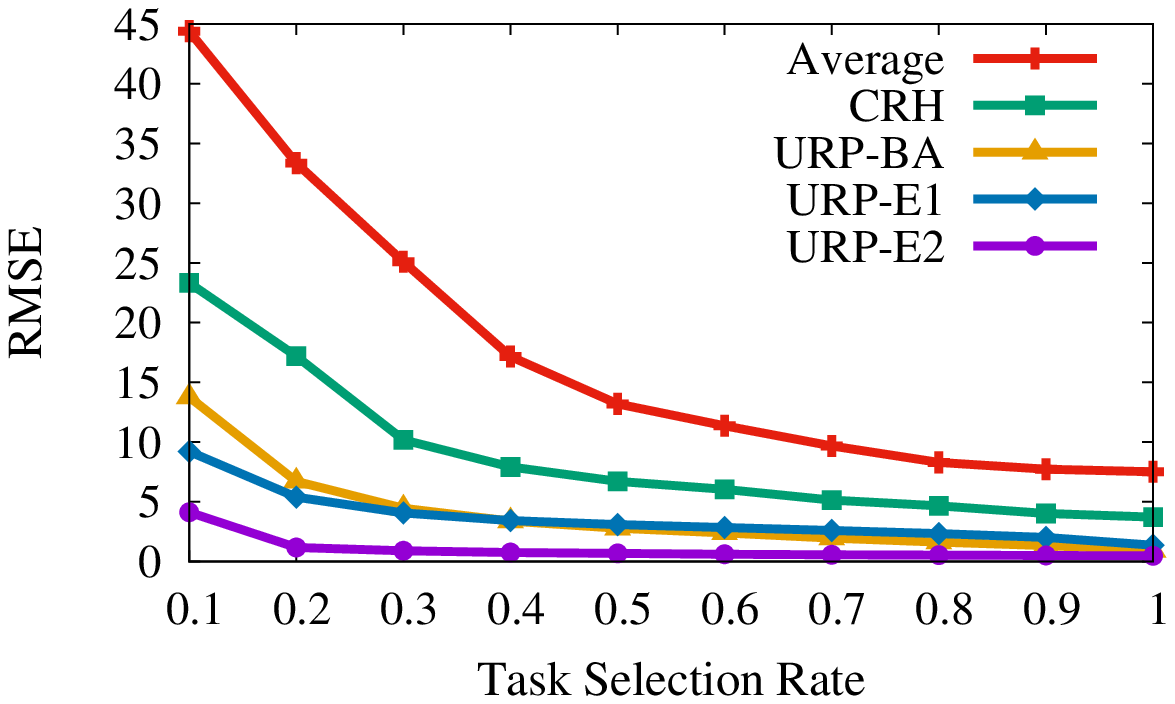}}
\vspace{-0.3cm}
\caption{Comparisons on estimation accuracy with varying users' task selection rate}
\label{varying-selection-rate}
\vspace{-0.5cm}
\end{figure*}

The performance of task matching strategies is measured on two different perspectives, \ie, task acceptance ratio and estimation accuracy. The task acceptance ratio is defined as the percentage of the recommended tasks that the users have selected, and the estimation accuracy is measured using RMSE or Error Rate depending on the data types of the tasks. The performance comparison of different task matching strategies is presented in Fig. \ref{task-matching-cmp}. We can see that the preference-only strategy has the highest task acceptance ratio, while the reliability-only strategy outputs the most accurate estimation results. That is because these two strategies match tasks to the users with the tendency of facilitating the match of one certain perspective. Comparing with other task matching strategies, we can see that our proposed hybrid recommendation strategy can achieve a good balance between the acceptance ratio and the estimation accuracy.

\subsection{Evaluations on A Large-Scale Scenario}

In this subsection, we examine the performance of our user profiling algorithm on a large-scale crowdsensing scenario.

In our simulation, there are 100 users and 1000 tasks. These tasks are randomly distributed among 20 categories. Each user's task selection rate is set to 10\%, \ie, each user contributes data to each task with \% probability. The ground truth of each task is randomly distributed within [30,100]. For each user $i$, if she contributes data to the task $j$ of category $c$, then her data $x_{i,j}$ is generated based on a Gaussian distribution with the mean $\hat{x}^*_j$ and variance $\frac{2}{q_{i,c}}$, \ie, $x_{i,j} \sim \mathcal{N}(\hat{x}^*_j,\frac{2}{q_{i,c}})$. In URP-E2, we randomly choose 1\% of tasks, and incorporate their ground truths in the user reliability profiling process.


In the simulation, we classify the users into three groups: reliable users, normal users, and unreliable users, where the users' reliability distributions in these three groups are $\mathcal{N}(0.75,0.1)$, $\mathcal{N}(0.5,0.1)$, and $\mathcal{N}(0.25,0.1)$, respectively. We consider three different settings. In the first setting, the users are classified into the three groups randomly. In the second setting, each user has 60\% probability of being classified into reliable users, 30\% normal users, and 10\% unreliable users, while in the third setting, each user has 10\% being reliable, 30\% being normal, and 60\% being unreliable. We assume that for each user, if her reliability for certain task is below 0.2, then the user will have 50\% probability of failing the task.

Fig. \ref{varying-num-users} presents the estimation accuracy of different algorithms with a varying number of the users. The number of users varies from 10 to 100 with the increment of 10. We can see that the simple average has the worst estimation accuracy, while URP-E2 achieves the lowest RMSE in all the three settings. In \ref{varying-num-users}(c), we observe that the RMSE first grows as the number of users increases, and then decrease when the number of users is getting larger. This is because that when the number of users is small, slightly increasing the number of users, especially unreliable users, may bring extra errors to the estimation results. As the number of users increases, the platform can access to more information, and thus can reduce the estimation errors.

Fig. \ref{varying-selection-rate} shows the estimation accuracy of different algorithms with varying task selection rate. We increase the task selection rate from 0.1 to 1 with the increment of 0.1. It can be seen that our proposed user profiling algorithm achieves the lowest RMSE, indicating the effectiveness of our algorithm. Besides, we can observe that the RMSE decreases as the task selection rate increases. This is because that increasing the task selection rate usually means having more data, \st, the platform can identify the users' reliability levels more accurately. A similar phenomenon was also observed in \cite{li2014confidence}.

We also examine the effect of the number of incorporated ground truths on the estimation accuracy. The results are shown in Fig. \ref{training-ratio}. We can see that having more truth can improve our estimation results. Besides, comparing the different settings, we can see that Setting 2 achieves the best estimation accuracy, since most users in Setting 2 are reliable.

\begin{figure}[!t]
\centering
\includegraphics[width=6cm]{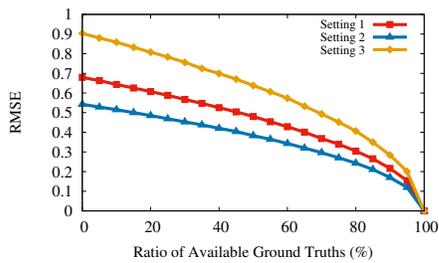}
\vspace{-0.2cm}
\caption{The Effect of Available Truth on Estimation Accuracy}
\vspace{-0.4cm}
\label{training-ratio}
\end{figure}

\section{Discussion}
\label{discussion}

In this section, we discuss several practical issues and potential extensions of our proposed personalized task recommendation methods.

\textbf{New User Problem:} For a user that is new to our system, we may have very little information (browsing history and data contribution) on the user. In this case, it is difficult to get an accurate preference or reliability profile of the user. Fortunately, this new user problem has been widely studied in traditional recommender system literature, \eg, \cite{rashid2002getting, yu2004probabilistic}, where their ideas can also be applied in our problem scenario. For example, we can recommend the most informative tasks to the new user, so as to gain knowledge of the user's preference and reliability. Heuristics may include random recommendation, recommending most popular tasks, and recommending tasks among different categories.

\textbf{Content-Based Reliability Prediction}: In this work, we propose a matrix factorization method to predict the missing entries in each user's reliability estimation. This method is able to capture inherent subtle characteristics of the users' reliability without the need of extracting features of the users and the tasks. However, one drawback of the approach is that we may not be able to interpret what factors influence the users' reliability. In situations where interpretability matters, using content-based methods such as building a classification model to predict the users' reliability can be a good alternative.

\textbf{General Reliability Profiling Problem}: In our reliability profiling model, we assume that each task only belongs to one category and tasks in different categories are independent. Sometimes, these assumptions may not hold. In these cases, a general reliability profiling problem can be considered, \ie, given the users' contributed data, we aim to estimate the unknown ground truths $\{x_j^* | j\in \mathbb{S}\}$, each user's reliability vector $\boldsymbol{q}_i \in \mathbb{R}^C$, and each task's weight vector $\boldsymbol{v}_j \in \mathbb{R}^C$, where $C$ is a hyperparameter determining the dimension of the vectors.

Having estimated $\boldsymbol{q}_{i}$ and $\boldsymbol{v}_j$, each user $i$'s reliability for each task $j$ can be calculated as $q_{i,j} = \boldsymbol{q}_i^T \cdot \boldsymbol{v}_j$. We can see that the model we proposed in Section \ref{urp-model} is a simplified version of the general reliability profiling problem, where we specify $C$ to be the number of task categories, and each entry $q_{i,c}$ of the task vector is 1 if task $j$ belongs to category $c$, and zero otherwise. Note that in the general reliability profiling problem, we now have three sets of unknown variables that need to be estimated, which is more difficult. We tend to leave this problem to our future work.

\textbf{Context-based User Profiling}: We observe that a user's preference and reliability can be dependent on the contextual situation of the user. For example, a user's preference may be dependent on time (\eg, time of a day, or season of the year) \cite{adomavicius2011context}. Also, as pointed out in \cite{liu2017context}, a user's reliability could also be influenced by her activity (\eg, sitting, walking, or running) and her surrounding environment (\eg, home, office, shopping mall). Thus, we think the problem of profiling the users' preference and reliability in a context-aware situation could also be an interesting future work.

\textbf{Other Task Matching Models}: Though this work focuses on the user-centric model, the estimated preference and reliability parameters of the users can be readily integrated into platform-centric model, allowing the platform to make centralized decisions based on them, such as selecting most reliable users to perform a sensing task \cite{he2015high}, determining the payments of the users based on their reliability levels \cite{peng2015pay,han2016posted,yang2017designing}, or integrating the preference or reliability information into the design of incentive mechanisms \cite{jin2015quality}.

\section{Related Work}
\label{related}
\textbf{Crowdsensing Applications}: The concept of mobile crowdsensing has attracted broad attention from both industry and academia, and has been applied in various application domains, including but not limited to environment monitoring \cite{mun2009peir,lu2009soundsense,gao2016mosaic}, indoor localization \cite{azizyan2009surroundsense,rai2012zee}, indoor floorplan construction \cite{alzantot2012crowdinside,gao2014jigsaw}, traffic and navigation \cite{zhou2012long,shu2015last}, and image sensing \cite{wang2014smartphoto}.

\textbf{Platform-Centric Crowdsensing}: Many researchers have studied the user selection problem in mobile crowdsensing. They usually modelled the problem from a game-theoretical perspective like \cite{yang2012crowdsourcing}. For example, Zhao \et \cite{zhao2014crowdsource} considered the problem of budget feasible mechanism design for crowdsensing, and proposed mechanisms for both offline and online scenarios. Karaliopoulos \et \cite{karaliopoulos2015user} addressed the user recruitment problem for opportunistic network scenario, and proposed two efficient algorithms to maximize the overall location coverage. Zhang \et \cite{zhang2016incentive} proposed a double auction mechanism for proximity-based mobile crowdsensing. In these works, the platform's main concern was to determine the set of selected users and their corresponding payments so as to maximize a certain optimization metric. They only considered the heterogeneity of the users and assumed that the tasks are of no differences. Some researches have also studied the task assignment problem in mobile crowdsensing. For example, He \et \cite{he2014toward} studied the optimal task allocation problem for location-dependent crowdsensing. Zhao \et \cite{zhao2014fair} considered the task allocation problem in crowdsensing with the objective of optimizing the energy efficiency of smartphones. Cheung \et \cite{cheung2015distributed} considered the distributed task selection problem for time-sensitive and location-dependent tasks. However, these works were all based on a platform-centric model. Besides, none of these work took the issue of data quality into consideration.

\textbf{User-Centric Crowdsensing}: Few researches have studied the user-centric model in crowdsensing. Karaliopoulos \et \cite{karaliopoulos2016first} adopted logistic regression techniques to estimate a user's probability of accepting a task, and tend to match tasks to users based on the information. However, they did not consider the users' data quality or reliability in performing the sensing tasks. Although Jin \et \cite{jin2015quality} and Han \et \cite{han2016posted} considered the problem of quality-aware task matching, they were based on the platform-centric model, and were unable to recommend personalized tasks for the users. In contrast, our work considers a user-centric task matching model by taking both the users' preference and data quality into consideration. A preliminary version of this work appears at INFOCOM 2018 \cite{yang2018towards}, while this work has substantial revision over the previous one including additional technical materials and discussions.

\textbf{Truth Discovery}: The problem of truth discovery has been widely studied to handle the situation where data collected from multiple sources tend to be conflicting and the ground truths are unknown \cite{yin2008truth}. Wang \et \cite{wang2012truth} considered the problem of truth detection in social sensing based on EM algorithm. Wang \et \cite{wang2013recursive} proposed a truth discovery algorithm to handle streaming data. Ouyang \et \cite{ouyang2014truth} proposed a truth discovery method to detect spatial events based on a graphical model. Su \et \cite{su2014generalized} designed a generalized decision aggregation framework for distributed sensing scenarios. Wang \et \cite{wang2014towards} studied the truth discovery problem in cyber-physical systems. Wang \et \cite{wang2015scalable} further exploited the problem of truth discovery for interdependent phenomena in social sensing. Meng \et \cite{meng2015truth} exploited the spatial correlations to improve the estimation accuracy. CRH \cite{li2014resolving} is a general truth discovery framework that can handle both continuous and categorical data. Li \et \cite{li2014confidence} considered truth discovery problem for long-tail data, and proposed a confidence-aware approach. Peng \et \cite{peng2015pay} propose an EM algorithm to quantity the users' data qualities in mobile crowdsensing. However, all of these works are based on unsupervised learning models, and thus may suffer from the initialization problem when most data are inaccurate \cite{li2016survey}. Yin and Tan \et \cite{yin2011semi} proposed a semi-supervised learning model to identify true facts from false ones. However, their work tended to focus on the truth estimation part, but did not output the reliability levels of the data sources, thus cannot address the need of user reliability profiling.

\textbf{Recommender System}: Recommender system has been a hot topic in recent decades. Generally, recommendation techniques can be classified into the following three categories: content-based recommendation, collaborative filtering-based recommendation, and hybrid recommendation \cite{adomavicius2005toward}. Besides various recommendation techniques, many practical issues in recommender systems have also been widely studied, including exploiting implicit feedback \cite{oard1998implicit,rendle2009bpr}, addressing negative feedback \cite{carroll1993explicit,lee2009reinforcing}, context-aware recommendation \cite{adomavicius2011context}, group recommendation \cite{masthoff2011group}, and so on. Nevertheless, these works only focused on the users' preferences, without considering the users' reliability. In contrast, in mobile crowdsensing, the users' reliability plays an important role in the effectiveness of the system, and thus should be taken into account in recommending tasks. To that end, we extend the traditional recommender systems by taking the users' reliability into the consideration and proposing to recommend tasks based on both the users' preference and reliability.

\section{Conclusion}
In this paper, we have studied the problem of personalized task matching in mobile crowdsensing. We have proposed a personalized task recommender framework that can recommend tasks to users based on a fine-grained characterization on both the users' preference and reliability. We have proposed methods to measure each user's preferences and reliability of different tasks, respectively. In particular, the proposed user reliability profiling algorithm originates from truth discovery problem, but surpasses existing truth discovery algorithms in three ways, \ie, by proposing a fine-grained multi-dimensional reliability profiling model, by exploiting the information of failed tasks, and also by incorporating a small number of ground truths to improve the estimation accuracy. Further more, we proposed a matrix factorization method to address a critical limitation of the existing truth discovery algorithms in estimating the users' reliability for the uninvolved tasks. Both a real-world experiment and a large-scale simulation have been conducted to evaluate our proposed methods. The evaluation results have demonstrated the good performance of our methods.

\bibliographystyle{IEEEtran}
\bibliography{bib}

\end{document}